\documentclass{article}
\usepackage{ijcai25}

\usepackage{times}
\usepackage{soul}
\usepackage{url}
\usepackage[hidelinks]{hyperref}
\usepackage[utf8]{inputenc}
\usepackage[small]{caption}
\usepackage{graphicx}
\usepackage{amsmath}
\usepackage{amsthm}
\usepackage{booktabs}
\usepackage{algorithm}
\usepackage{algorithmic}
\usepackage[switch]{lineno}
\usepackage{xfrac}
\usepackage{xspace}

\usepackage{latexsym}

\newtheorem{theorem}{Theorem}
\title{Equitable Mechanism Design for Facility Location}

\author{
    Toby Walsh
    \affiliations
    AI Institute, UNSW Sydney
    \emails
    tw@cse.unsw.edu.au
}

\newcommand{\myOmit}[1]{}

\newcommand{\mymin}{\mbox{\rm min}}
\newcommand{\mymax}{\mbox{\rm max}}

\newcommand{\myendpoint}{\mbox{\sc EndPoint}\xspace}
\newcommand{\myendpointgamma}{\mbox{$\mbox{\sc EndPoint}_{\gamma}$}\xspace}

\newcommand{\mymidornearest}{$\mbox{\sc MidOrNearest}$\xspace}
\newcommand{\mylrm}{$\mbox{\sc Lrm}$\xspace}
\newcommand{\mylrmt}{$\mbox{\sc LrmT}$\xspace}
\newcommand{\mymedian}{$\mbox{\sc Median}$\xspace}
\newcommand{\myleftmost}{$\mbox{\sc Leftmost}$\xspace}

\begin{document}

%%% Use this combinations of commands to specify all authors of your 
%%% paper. Use \fnms{} and \snm{} to indicate everyone's first names 
%%% and surname. This will help the publisher with indexing the 
%%% proceedings. Please use a reasonable approximation in case your 
%%% name does not neatly split into "first names" and "surname".
%%% Specifying your ORCID digital identifier is optional. 
%%% Use the \thanks{} command to indicate one or more corresponding 
%%% authors and their email address(es). If so desired, you can specify
%%% author contributions using the \footnote{} command.

\maketitle

\begin{abstract}
  We consider strategy proof mechanisms for facility location
  which maximize equitability between agents.
  As is common in the literature, we measure equitability with the
  Gini index. 
  We first prove a simple but fundamental impossibility result
  that no strategy proof mechanism can bound the approximation ratio
  of the optimal Gini index of utilities for one
  or more facilities.
  We propose instead computing approximation ratios of the
  complemented Gini index of utilities, and consider how well both deterministic and randomized mechanisms
  approximate this. 
\myOmit{  For deterministic mechanisms and a single facility,
  we prove that the \mymedian\ mechanism 2-approximates this
  objective, but that the \mymidornearest\ mechanism does even better,
  providing
  a $\sfrac{6}{5}$-approximation. This is optimal
  amongst  strategy proof, anonymous and Pareto efficient
  mechanisms. 
  For randomized mechanisms and a single facility,
  we prove that the \mylrm\ mechanism (which is optimal
  with respect to approximating the maximum distance)
  is not optimal with respect to approximating  equitability.
  We also extend these approximability
  results to multiple facilities.
  For instance, we propose a new mechanism for locating two
  facilities with an approximation ratio
  better than the \myendpoint\  mechanism,
  the only deterministic mechanism which approximates
  well the minimum utility or maximum distance.}
  In addition, as Nash
  welfare is often put forwards as an equitable compromise
  between egalitarian and utilitarian outcomes, we consider
  how well %these
  mechanisms approximate the Nash welfare.
\end{abstract}

\section{Introduction}

Mechanism design is the problem of designing rules for a game to achieve a
specific outcome, even though each participant may be
self-interested.
The aim is to design rules so that the participants are incentivized to behave as the
designer intends. This typically includes achieving properties such as 
truthfulness, individual rationality, budget balance, and
maximizing social welfare. Here we consider another
desirable property that designers might look to achieve:
equitability. How does a mechanism designer ensure that all
participants are as equally 
happy with the outcome as is possible?
Surprisingly, equitable mechanism
design has received somewhat limited attention so far in the
scientific literature.

Central to this question of equitable mechanism
design is defining what it means
for an outcome to be equitable. Consider
the simple decision making problem of locating a facility
along a line. This models a number of real world problems
such as picking the room temperature for a classroom, or the deadline for
a project. 
Agents are supposed to have single peaked preferences, prefering the facility to be nearer to their location. 
And we look to design mechanisms which locate the
facility so that the distances which the different agents must travel are
as small and as similar as is possible. 
In general,
of course, the distances agents travel may have to be different.
Consider locating a facility on %the interval
$[0,1]$,
with three agents: one at 0, another at $\sfrac{1}{2}$ and the final
at 1. The agent at $\sfrac{1}{2}$ inevitably has to be nearer the facility
than at least one of the other two agents.
%There is no place on $[0,1]$
%that the facility can be located that ensures all three
%agents are the same distance from the facility.
Where then do we locate
the facility to ensure the most equitable outcome?
In this case, locating the facility at $\sfrac{1}{2}$
might seem best.
\myOmit{
  The agents at the two extremes
both have to travel an equal distance, while the
third agent is even better off. In addition,}
The maximum
distance any agent travels is the minimum it can be.
\myOmit{
Any other solution is more inequitable as one of the agents at the endpoints
must travels a greater distance, while the agent at the
other endpoint travels less. }

Our goal then is to design equitable mechanisms for
facility location in which agents are incentivized to report
sincerely. 
While our focus is on 
the facility location problem, there are some
general conclusions that can be drawn from this study.
First, designing mechanisms for equitability is
{\bf somewhat different} to
designing mechanisms for objectives such as social
welfare. For instance, equitability considers all agents, while egalitarian welfare
considers just the worst off agent. We will show that
mechanisms with approximate well the egalitarian welfare
may not return equitable outcomes.
For instance,
the randomized \mylrm\ mechanism (described shortly)
has good welfare properties but is less good at returning
equitable solutions. 
Second, designing strategy proof mechanisms for equitability is 
possible if we sacrifice a little optimality.
For example, we identify a {\bf strategy proof mechanism
for one facility with optimal equitability}. 
And by carefully considering
how existing mechanisms perform
poorly, we design
a {\bf
new strategy proof mechanism for two facilities
with close to optimal equitability}.
And third, a key component in designing such
new mechanisms with good performance is to 
{\bf censor extreme outcomes}. 
We conjecture
that mechanism design for equitability
may offer promise in other
domains such as fair division and ad auctions.

\section{Facility Location}

The facility location problem is a classic problem in
social choice and multiagent decision making
in which we need to decide where to locate a facility
to serve a set of agents
with single peaked preferences \cite{prefbook}.
The problem generalizes to
locating two (or more) facilities, in
which case
we suppose agents are served by the nearest facility. 
We consider $n$ agents located on the real line at $x_1$ to $x_n$  and
a facility at $y$.
Without loss of generality, we suppose $x_1 \leq \ldots \leq x_n$. 
We let $d_i$ be the distance of agent $i$ to the facility:
$d_i =  | x_i - y|$.
As in several previous studies, we assume agents and facilities
are limited to the interval
$[0,1]$, and the utility of agent $i$ is $u_i = 1-d_i$.
The interval could be $[a,b]$, % instead of $[0,1]$,
in which case we normalise by $b-a$.

Supposing agents and facilities lie on
an
interval 
is interesting for both practical and theoretical reasons. 
In particular, agents and facilities are often limited to a finite
interval and cannot be located outside those limits. 
For example, when setting a thermostat, we
have a temperature range limited by the boiler.
As a second example, when locating a water treatment plant on a river, the
plant must be on the river itself.
As a third example, when locating a shopping centre, 
the centre might have to be on the fixed (and finite) road network.
%There are thus many situations where agents are
%limited to a finite interval.
Restricting agents to a finite interval
also limits the extent to which
agents can misreport their location
to influence the outcome.
Several other recent works have used a finite interval
(e.g. \cite{cheng2013,fsaaai2015,flprevisit,abs-2111-01566,wpricai21,alswneurips2022,propwine2022,allrwaamas2023,lalrwaamas2024}). 

A deterministic mechanism $f$ locates the facility at a location $y$. 
Formally, $f(x_1,\ldots,x_n) = y$.
%We let $d_i$ be the distance of agent $i$ to the facility:
%$d_i =  | x_i - y|$.
%A randomized mechanism returns a probability distribution of
%facility locations.
A mechanism is {\em anonymous} iff any permutation of the
agents returns the same facility location.
Formally $f$ is anonymous iff for any
permutation $\sigma$,
$f(x_{\sigma(1)}, \ldots,x_{\sigma(n)})=f(x_1,\ldots,x_n)$.
A mechanism is {\em Pareto efficient} iff we cannot move
the location of the facility and make one agent better off (nearer to
the facility) and no agent worse off.
Formally $f$ is Pareto efficient iff
for any $x_1, \ldots, x_n$, 
there is no location $z$ and agent $i$ with
$|x_i-z|  < |x_i-f(x_1,\ldots,x_n)|$
and $|x_j-z|  \leq |x_j-f(x_1,\ldots,x_n)|$ for all $j \in [1,n]$. 
We limit our attention to {\em unanimous} mechanisms
that locate the facility where all agents agree.
Formally $f$ is unanimous iff for any
$x$, $f(x,\ldots,x)=x$.
This is a weaker condition than Pareto efficiency.
Unanimity rules out undesirable mechanisms such as the
mechanism which always locates the facility at $\sfrac{1}{2}$.
We consider how well 
mechanisms approximate some objective $O$
like the (soon to be defined) Gini index. 
For a maximization (minimization)
objective, the approximation ratio
is the maximum ratio of $O_{opt}/O_{approx}$ ($O_{approx}/O_{opt}$)
where $O_{opt}$ is the optimal %objective
value and
$O_{approx}$ is the approximately optimal %objective
value
returned by the mechanism. 
%For instance, a mechanism 2-approximates an objective
%iff the approximate solution returned by the mechanism is always within a
%factor of 2 of the optimal. 

One of our major concerns is strategic manipulation.
Agents are self-interested so we look for mechanisms where
agents cannot improve their outcome by mis-reporting their
location.
A mechanism is {\em strategy proof} iff no agent can misreport
their position and reduce their distance to the nearest facility.
Formally $f$ is strategy proof iff for any
$x_1,\ldots,x_n$, and agent $i$, it is not the
case that there exists $x_i'$ with
$|x_i - f (x_1,\ldots,x_i',\ldots,x_n)| < |x_i-f(x_1,\ldots,x_i,\ldots,x_n)|$. 
Moulin
\shortcite{moulin1980}
has provided an elegant characterization of
strategy-proof mechanisms:
a mechanism
is anonymous, strategy-proof and Pareto efficient 
iff it is the 
median rule (which locates the facility at the
median agent) with at most $n-1$ phantoms (additional
``agents'' reporting fixed locations). 
We consider various strategy proof mechanisms from
previous studies of facility location.
%We also design
%several new mechanisms optimizing equitability. 
 The \myleftmost\ mechanism locates the facility at $x_1$, the leftmost agent.
 This is equivalent to the median rule with $n-1$ phantoms
  at $0$. 
 The \mymedian\ mechanism locates the facility at $x_{\lceil \sfrac{n}{2}
  \rceil}$, the median agent. 
This is equivalent to the median rule with $\lfloor \sfrac{n}{2} \rfloor$ phantoms
at $0$, and the rest of the %$n-1$
phantoms at $1$.
 The \mymidornearest\ mechanism locates the facility at $x_n$ when $x_n
< \sfrac{1}{2}$,
at $\sfrac{1}{2}$ when $x_1 \leq \sfrac{1}{2} \leq x_n$,
and at $x_1$ when $x_1 > \sfrac{1}{2}$.
 This is equivalent to the median rule with $n-1$ phantoms
 at $\sfrac{1}{2}$.
 When locating two facilities,
 the \myendpoint\ mechanism locates one facility
at $x_1$ and another at $x_n$. 
These mechanisms are all deterministic.
We also consider  randomized mechanisms which
return a lottery over solution, and the expected
distance of agents from the facility. 
For example, the strategy
proof \mylrm\ mechanism (Left, Right, or Midpoint)
locates the facility at $x_1$ 
with probability $\sfrac{1}{4}$, at $\sfrac{(x_1+x_n)}{2}$
with probability $\sfrac{1}{2}$, and at $x_n$
with the remaining probability $\sfrac{1}{4}$
\cite{ptacmtec2013}.
%Agents cannot reduce their
%expected distance to the facility by misreporting their location.

While our results are, like many previous studies,
focused on the 1-d setting, they
are interesting more broadly. 
The 1-d facility location problem models several real world problems
such as locating
distribution centres along a highway. There are
also non-geographical settings that are
1-d %problems
(e.g. setting a thermostat or tax rate). In addition,
we can solve more complex problems in higher dimensions
by decomposing them into 1-d problems. 
Finally, the 1-d problem is the starting point to
consider more complex metrics such as trees and networks.
For instance, lower bounds for 1-d can be inherited
for 2-d and other metric spaces.

\section{Minimizing the Gini Index}

The Gini index 
is one of the most widely used measures
of equitability. % in economics.
\myOmit{It can be justified axiomatically in a number of
ways. For instance, it is the unique index that
satisfies scale invariance, symmetry, proportionality and convexity in similar rankings.}
Unsurprisingly 
it has been used in facility location problems. 
For example, Mulligan \shortcite{equiflp} argues that simple equitability measures 
like maximum distance ignore the distribution of distances
and recommends measures like the Gini index. 
The Gini index of
distances is: % defined by:
\begin{eqnarray*}
G_d & = & \frac{\sum_{i\leq n} \sum_{j \leq n}  | d_i - d_j |}{2n
          \sum_{i \leq n}   d_i}
\end{eqnarray*}
This lies in $[0,1]$, takes the value 0 for an equitable solution
when $d_i=d_j$ for all $i$ and
$j$, and increases in value as distances become more unequal. 
%Our results here consider exclusively the Gini index but
%other measures of equitability besides the Gini index
%behave similarly (such as the Atkinson index, Hoover or Robin Hood
%index). 
If all agents are at the same location, then any facility location
is an equitable solution since all agents travel the same
distance. Therefore equitability alone is not sufficient to
guarantee solutions are desirable.
We might also demand additional properties
like unanimity.

But does minimizing the Gini index of distances even guarantee an
equitable and favourable outcome?
Minimizing the Gini index of distances favours distances being
large over distances being small. As an example, 
suppose we have $n$ agents at $\sfrac{1}{2}$,
and  one agent at both 0 and $1$, with $n>2$. 
Locating the facility at 0 or 1,
  gives the minimum Gini index of distances
  of $\frac{2(n+1)}{(n+2)^2}$. 
In this solution, $n$ of the agents have
to travel a distance of $\sfrac{1}{2}$ and one agent
has to travel the maximum distance of 1.
A more equitable solution has the facility at $\sfrac{1}{2}$,
with most of the agents travelling no distance at all, and
just two agents travel a distance of $\sfrac{1}{2}$.
Perversely this
has a larger Gini index of distances of $\frac{n}{(n+2)}$. 
% d, n 1@1/2 0, 1@1 1/2 1@0 1/2, f@1/2
% G = 1/3 n/(n+2)  
%
% d, n 1@1/2 1/2, 1@1 1, 1@0 0, f@0
% G = 4/9  2(n+1)/(n+2)(n+2)
%
%We conclude therefore that minimizing the Gini index of
%distances in facility location does not ensure equitability. 

We propose instead to {\bf consider utilities rather than
distances} by minimizing the Gini index of utilities:
\begin{eqnarray*}
G_u & = & \frac{\sum_{i \leq n} \sum_{j \leq n}  | u_i - u_j |}{2n
          \sum_{i \leq n}   u_i}
\end{eqnarray*}
Minimizing the Gini index of utilities prefers most agents
having
a large utility over most agents having a large distance to 
travel. Returning to
the previous example, minimizing the Gini index of
utilities locates the facility at $\sfrac{1}{2}$ which is, 
as we argued, the most equitable facility location.
%
%
% u, n 1@1/2 1, 1@1 1/2 1@0 1/2, f@1/2
% G = n/(n+2)(n+1)  n=2, G=1/6, n=3, G=1/4
%
% u, n 1@1/2 1/2, 1@1 0, 1@0 1, f@0
% G = 4/9  2(n+1)/(n+2)(n+2)   n=2, G=3/8, n=3, G=8/25
%
%
%
%
  \myOmit{
\begin{figure}[htb]
\hspace{-6em} \includegraphics[width=0.7\textwidth,height=210pt]{robots.pdf}
 \caption{Example of facility location problem
   with $n$ agents illustrating why minimizing the Gini index of
      utilities ($G_u$) rather than the Gini index of distances ($G_u$) gives an
      equitable outcome in which agents travel the least.}
\end{figure}
}
{Note that minimizing the Gini index %of utilities
is also compatible
with Pareto efficiency. Whilst not all solutions 
that minimize the Gini index of utilities are Pareto efficient,
there is always a solution minimizing the
Gini index that is} Pareto efficient. 

Our first result is {\bf an impossibility}.
No strategy proof mechanism has a bounded
approximation ratio. % of 
%the Gini index. % of utilities. 

\begin{theorem}
No strategy proof mechanism 
for locating one or more facilities on $[0,1]$
has a bounded approximation ratio 
for the Gini index of utilities.
%with any number of agents $n > 1$. 
\end{theorem}
\begin{proof}
Suppose there exists a strategy proof mechanism for two agents with a bounded approximation ratio for locating a single
facility. 
%For $n=2$,
Consider $x_1=0$ and $x_2=\sfrac{1}{2}$.
%The Gini index takes the value zero iff the facility is
%located at the midpoint, $\sfrac{1}{4}$. 
To have a bounded approximation ratio, the facility must be 
at the optimal location, $\sfrac{1}{4}$. %this position. 
Consider $x_1=0$ and $x_2=1$. 
To have a bounded approximation ratio, the facility must be 
at $\sfrac{1}{2}$. 
Hence,
the agent at $\sfrac{1}{2}$ in the first
scenario had an incentive to mis-report
their location as 1. %This contradicts our assumption that 
%there exists a strategy proof mechanism 
%with a bounded approximation ratio. 
A similar construction can be made for more
agents and facilities.
\myOmit{Suppose there exists a strategy proof mechanism 
with a bounded approximation ratio for locating $n$ facilities ($n > 2$). 
Consider agents at $0$, $\sfrac{1}{16n}$ and $\sfrac{3}{2n}$,
$\sfrac{5}{2n}$, \ldots, $1 - \sfrac{1}{2n}$. 
The Gini index takes its zero value when 
facilities are located at $\sfrac{1}{32n}$, and $\sfrac{3}{2n} \pm \sfrac{1}{32n}$,
$\sfrac{5}{2n} \pm \sfrac{1}{32n}$, \ldots, $1 - \sfrac{1}{2n} \pm \sfrac{1}{32n}$.
To satisfy the approximation ratio, the mechanism must locate the 
$n$ facilities at these $n$ locations. 
Suppose the second agent reports location $\sfrac{1}{8n}$. 
To satisfy the approximation ratio, the leftmost facility must now be located
at $\sfrac{1}{16n}$. 
Hence, if agents are at 0, $\sfrac{1}{16n}$ and 
$\sfrac{3}{2n}$, $\sfrac{5}{2n}$, \ldots, $1 - \sfrac{1}{2n}$ then
the agent at $\sfrac{1}{16n}$ has an incentive to mis-report
their location as $\sfrac{1}{8n}$. This contradicts our assumption that 
there exists a strategy proof mechanism 
with a bounded approximation ratio. }
\end{proof}

Randomization does not help escape this impossibility.
The proof works whether mechanisms are deterministic or randomized. 
%A similar argument also shows that no strategy proof
%mechanism has a bounded approximation ratio of the Gini
%index of distances. 
The problem with the approximation ratio of the Gini
index %, whether it be the Gini index of utilities or of distances, 
is that %the approximation
% ratio of the Gini index
this focuses on equitable problems where 
the index is zero (and all
distances/utilities are equal)
%For example, with two agents, we must
%locate the facility close to their midpoint. 
%It is easy to see that this 
%is incompatible with strategy proofness.
A natural way around this problem %of unbounded
%approximation ratios
is to consider the complemented Gini index (that is,
$1-G$). This is again in $[0,1]$. It is 1 when utilities (or
distances) are equal,
and becomes smaller as problems become more inequitable. 
Our goal now is to {\em maximize} the complemented Gini index. % rather
%than {\em minimize} the original Gini index.
%
Considering the approximation ratio of the complemented Gini index
switches focus %away from approximating well
%equitable problems in which utilities (or
%distances) are balanced to approximating well
onto %xchallenging and
inequitable problems in which utilities (or distances) are
necessarily imbalanced (such as the earlier example 
with agents at 0, $\sfrac{1}{2}$ and $1$). 

You might be concerned that by shifting to the complement of the
  Gini index, we are just replacing one problem (approximating
  equitable problems with Gini indices close to zero) with another
  (approximating inequitable problems with Gini indices close to 1,
  and complemented Gini indices close to zero).
  This is not the case. While Gini indices can indeed be close to zero
  (and hard to approximate within a constant factor),  the {\em optimal} Gini index of
  utilities in facility location is never close to 1 and, as we will show, can be
  approximated well by strategy proof mechanisms.

% \section{Equitable Mechanisms for a Single Facility}
% \section{Deterministic mechanisms}
\section{One Facility}

We first demonstrate that 
there exist strategy proof mechanisms which approximate 
well the complemented Gini index. % of utilities.
\myOmit{It is easy to see that no strategy proof mechanism can return the
optimal Gini index, but must 
at best approximate it. For instance, with 
two agents, a mechanism that returns the optimal 
solution would locate the facility at the midpoint between the two
agents. This is not strategy proof.
}
We start with one of the simplest possible %strategy proof
mechanisms.
The \myleftmost\ mechanism is strategy proof and
2-approximates the maximum distance. This 
is optimal as no deterministic and strategy proof mechanism
can do better\footnote{Procaccia and Tennenholtz
\shortcite{ptacmtec2013} demonstrate this for
the real line. However, the result easily
extends to any fixed interval.}. However, the \myleftmost\ mechanism
does not return very equitable solutions. 
%While it has a bounded
%approximation ratio for the complemented Gini index of
%utilities, the bound is large (and increases with
%the number of agents). 

\begin{theorem}
For a facility location problem with $n$ agents, the \myleftmost\ mechanism
%is strategy proof and
$n$-approximates the complemented Gini index of utilities.
\end{theorem}
\begin{proof}
  If all agents are at the same location, then the
  \myleftmost\ mechanism
is optimal with respect to 
the complemented Gini index. Therefore we suppose
agents are at two or more locations.
The smallest possible complemented Gini index (which the \myleftmost\
mechanism achieves) is when one agent is at 0 and the remaining $n-1$ agents are
at 1. 
The complemented Gini index in this case is just $\sfrac{1}{n}$.
This compares to an optimal of 1 when the facility
is at $\sfrac{1}{2}$. 
%Note that it is impossible for the Gini index to be larger (and
%thus for the complemented Gini index to be smaller). 
%The optimal complemented Gini index that can be achieved
%in this case is, on the other hand,
%maximal. In particular, if we locate the facility at $\sfrac{1}{2}$, the 
%complemented Gini index is 1. The \myleftmost\ mechanism
%therefore $n$-approximates the optimal complemented Gini index. 
\end{proof}

The \mymedian\ mechanism does significantly better than the \myleftmost\
mechanism at returning equitable solutions.
This is unsurprising as the \mymedian\ mechanism 
tends to be more balanced than the \myleftmost\ mechanism which
necessarily locates the facility at an extreme location.
When considering the maximum distance agents must travel, we
cannot distinguish between the \myleftmost\ and \mymedian\
mechanisms. Both mechanisms
2-approximate the maximum distance.
Here we see that {\bf the Gini index distinguishes them apart}, suggesting that \mymedian\ is more equitable. 

\begin{theorem}
  % For a facility location problem,
  The \mymedian\ mechanism 
$2$-approximates the complemented Gini index of utilities.
\end{theorem}
\begin{proof}
We prove that the \mymedian\ mechanism always returns
a facility location that gives a complemented Gini index of utilities
of $\sfrac{1}{2}$ or greater.
%As the complemented Gini index
%is 1 or smaller,
Hence the \mymedian\ mechanism cannot be worse than a
2-approximation.
Indeed, if one agent is at 0 and another at 1, %we note that
it is at best a 2-approximation. 
%
%We now consider a more general setting. 
%Suppose agents are at $x_1$ to $x_n$ with $x_ 1 \leq \ldots \leq x_n$.
%We shift agents left so $x_1 =0$.

There are two cases.
In the first, $n=2k$. 
Note that the median agent is at $x_k$, which is the location of the
facility. We assume $x_k \leq
\sfrac{1}{2}$ otherwise we reflect the position of any agent $x$ onto
$1-x$. %(and again if necessary shift all agents left so $x_1=0$). 
We first prove that $\sum_i u_i \geq k$.
In fact, the sum equals $k$ when $x_1  = \ldots  = x_k = 0$ and
$x_{k+1} = \ldots = x_n = 1$. Suppose there is a smaller
%value
%for the
sum of utilities for some other $x_1'$ to $x_n'$. If we map $x_i'$
onto 0 for
$i \leq k$ then
each $u_i$ for $i \leq k$ increases less than each $u_i$ for $i > k$
decreases. That is, the sum of utilities would decrease which is a
contradiction. Hence, $x_1'=\ldots =
x_k'=0$. Similarly, suppose $x_{k+1}'<1$. Then mapping $x_i$ onto 1
for $i>k$ would decrease the sum of utilities which is again a
contradiction. Hence, the smallest sum of utilities occurs when
$x_1  = \ldots  = x_k = 0$ and $x_{k+1} = \ldots = x_n = 1$
and this sum is $k$. 

We next prove that $\sum_i \sum_j |u_i-u_j| \leq 2k^2$.
Again, the maximum %value of this
double sum is when $x_1  = \ldots  = x_k = 0$ and
$x_{k+1} = \ldots = x_n = 1$. 
We consider different terms in the double sum. 
If we consider the pair of terms $|u_i - u_j | + |u_{i} - u_{n-j+1}|$ for
$i < j \leq k$ then as $x_i$ and $x_j$ are at or to the left of
$x_k$, and at or the right of
$x_k$, it follows that the sum of these two differences equals or is
less than 1. A similar argument applies to the sum of terms
 $|u_{n-i+1} - u_{n-j+1} | + |u_{n-i+1} - u_{j}|$. Hence the $4k^2$
terms  have  a maximum sum of $2k^2$. 
This again occurs when 
$x_1  = \ldots  = x_k = 0$ and $x_{k+1} = \ldots = x_n = 1$. 
The complemented Gini index is $1 - \frac{\sum_i \sum_j | u_i -
  u_j|}{2n \sum_i u_i}$. The minimum value this takes is lower bounded
by the maximum value of $\sum_i \sum_j | u_i -  u_j|$ divided by the
minimum value of $\sum_i  u_i$. That gives a lower bound of
$1 - \frac{2k^2}{4k \cdot k}$ or $\sfrac{1}{2}$.

In the second case $n=2k+1$ is odd.  We suppose $x_{k+1} \leq \sfrac{1}{2}$.
This is the median agent and therefore location of the facility.
By a similar argument,
$\sum_i u_i$ takes a minimum value of $k+1$,
and $\sum_i \sum_j | u_i -  u_j|$ takes a maximum value of
$2k(k+1)$ when 
$x_1  = \ldots  = x_{k+1} = 0$ and $x_{k+2} = \ldots = x_n = 1$
The complemented Gini index takes its minimum value of
$1 - \frac{2k(k+1)}{2(2k+1)(k+1)}$ or $1-\frac{k}{2k+1}$
which tends to 
  $\sfrac{1}{2}$ from above as $k$ goes to infinity. 
\end{proof}

Can we do even better than this? Yes, we can. Consider
the \mymidornearest\ mechanism. 
This is strategy proof, 2-approximates the maximum
distance
(recall that
no strategy proof and deterministic mechanism can do better),
%\cite{ptacmtec2013}),
and $\sfrac{3}{2}$-approximates the minimum utility
(no strategy proof and deterministic mechanism can
again do better \cite{wecai2024}).
We now show that the \mymidornearest\ mechanism
also approximates well
the complemented Gini index of utilities. 

\begin{theorem}
The \mymidornearest\ mechanism $\frac{(n^2+n)}{(n^2+1)}$-approximates the optimal
complemented  
Gini index of utilities for $n$ agents ($n \geq 1$).
The worst case is $n=2$ or $n=3$ when it provides a 
$\sfrac{6}{5}$-approximation. 
\end{theorem}
\begin{proof}
% one agent at 1/2
% n-1 agents at 1
% u1 = 1
% ui = 1/2 for i>1
% G = (n-1)/2n(1+(n-1)/2) = (n-1)/n(n+1) >= 1/6 for n=2
  For $n=1$, the \mymidornearest\ mechanism
  locates the facility at the single agent which is optimal. 
%  as required. 
For $n \geq 2$, observe that the
\mymidornearest\ mechanism guarantees that $u_i \geq \sfrac{1}{2}$ for
any $i$. 
The most inequitable outcome returned %by this mechanism
then is
when $u_i=\sfrac{1}{2}$ for $i < n$ and $u_n=1$.
This occurs, for example, when $x_i = 0$ for $i<n$ and
$x_n = \sfrac{1}{2}$,
and the %\mymidornearest\
mechanism locates the
facility at $\sfrac{1}{2}$. This gives a
%Gini index of utilities of $\frac{(n-1)}{(n^2+n)}$,
%and thus a
complemented Gini index of utilities of $
\frac{(n^2+1)}{(n^2+n)}$.
% 1 - (n-1)/n(n+1) = (n^2 + n -n +1)/n(n+1) = (n^2+1)/n(n+1) =
  Coincidently, when $x_i = 0$ for $i<n$ and
  $x_n = \sfrac{1}{2}$, there is an optimal and perfectly equitable solution
  which locates the facility at $\sfrac{1}{4}$, giving a
  complemented Gini index %of utilities
  of 1. 
  Hence, the most inequitable outcome for the \mymidornearest\
  mechanism occurs when there is
  an optimal and perfectly equitable
  solution.   The approximation ratio of the
  \mymidornearest\ mechanism is thus $\frac{(n^2+n)}{(n^2+1)}$. 
This ratio is maximized for $n=2$ or $n=3$
when it provides a 
$\sfrac{6}{5}$-approximation. 
%Note that as $n$ approaches infinity, the approximation
%ratio approaches 1 from above. 
\end{proof}

%\section{Lower Bound on Approximability}

In fact, {\bf we cannot do better} than a $\sfrac{6}{5}$-approximation.
%Any anonymous,
%Pareto efficient and strategy proof mechanism
%provides at best
%a $\sfrac{6}{5}$-approximation. 

\begin{theorem}
  A deterministic mechanism 
  that is anonymous, Pareto efficient
  and strategy proof at best
  $\sfrac{6}{5}$-approximates the complemented Gini index of
utilities. \end{theorem}
\begin{proof}
Such a mechanism is a median mechanism with $n-1$ phantoms.
Consider $n=2$, the phantom at $a$,
and one agent at 0, and another at 
$a$.
Suppose $a \leq \sfrac{1}{2}$.
The case of $a \geq \sfrac{1}{2}$ is dual. 
The median mechanism locates the facility at $a$. 
The complemented Gini index of utilities is
$\frac{(4-3a)}{(4-2a)}$. 
  % 0, u 1-a  
  % a, u 1
  % G = a/2(2-a), CG= (4-3a)/(4-2a)        
This compares to an optimal of 1. 
Therefore the approximation ratio is
$\frac{(4-2a)}{(4-3a)}$ which is in the
interval $[1,\sfrac{6}{5}]$ for $0 \leq a \leq \sfrac{1}{2}$. 
Now consider one agent at $a$, and another at 
1. This median mechanism again locates the
facility at $a$.
The complemented Gini index of utilities is
now $\frac{(1+3a)}{(2+2a)}$. 
% a, u 1
% 1, u a
  % G = (1-a)/2(1+a), CG= (1+3a)/(2+2a)
This again compares to an optimal of 1. 
Therefore the approximation ratio is
$\frac{(2+2a)}{(1+3a)}$ which is in the
interval $[\sfrac{6}{5},2]$ for $0 \leq a \leq \sfrac{1}{2}$.
Over the two scenarios, the best % approximation
ratio that can be achieved is $\sfrac{6}{5}$. 
\end{proof}

% 0
% 1/2
% fac a
% u 1-a
% u 1/2+a
% G = (2a-1/2)/3 = (4a-1)/6, CG=(7-4a)/6 = (7- 8/5)/6 = 9/10 for a=2/5 
% alpha = 6/(7-4a)

% 0 
% 1
% fac b
% u 1-b
% u b
% G = (1-2b)/2. CG=(1+2b)/2 = 9/5 / 2 = 
% alpha = 2/(1+2b)

% 2/(1+2a) = 6/(7-4a)
% 14-8a = 6+12a
% 20a = 8
% a = 2/5
% alpha = 2/9/5 = 10/9

\myOmit{
We can relax the assumption of Pareto efficiency, but this 
impacts slightly the lower bound that we are able to prove. 

\begin{theorem}
  No deterministic
and strategy proof mechanism
  can do better than $\sfrac{10}{9}$-approximate the complemented Gini index
  of utilities.
\end{theorem}
\begin{proof}
Consider an agent at 0 and $\sfrac{1}{2}$.
The optimal Gini index is $0$ with the
facility at $\sfrac{1}{4}$. Suppose there is a deterministic
and strategy proof mechanism with a better approximation
ratio. Then
the facility must be to the left of $\sfrac{2}{5}$
as the complemented Gini index
is $\sfrac{9}{10}$
when the facility is at $\sfrac{2}{5}$, and the optimal
is $\sfrac{10}{9}$ times larger. If the
facility is to the right of $\sfrac{2}{5}$,
the complemented Gini index is even less,
and fails to achieve the approximation ratio. 
Suppose the agent at $\sfrac{1}{2}$ reports their location
as 1. The optimal Gini index given this report
is again 0 with the
facility at $\sfrac{1}{2}$. To achieve the
approximation ratio of 
the complemented Gini index, the facility must now be in the
interval $(\sfrac{2}{5},\sfrac{3}{5})$. But this puts the facility
strictly closer to $\sfrac{1}{2}$ contradicting the assumption
that the mechanism is strategy proof. 
\end{proof}
}

\section{Two Facilities}

%We next consider how these results extend to multiple facilities.
With two facilities, 
the \myendpoint\ mechanism is the only strategy proof and
deterministic mechanism on the real line
with a bounded approximation ratio of the optimal maximum distance \cite{ft2013}. % or
%of the optimal minimum utility. 
%More precisely, it 2-approximates the optimal maximum distance, and 
%$\sfrac{3}{2}$-approximates
%the optimal minimum utility. 
With respect to the complemented Gini index of utilities, it provides
a reasonable
approximation ratio of the optimal. However, somewhat
surprisingly, it {\bf does not offer the best possible ratio}
amongst
strategy proof and deterministic mechanisms. 

\begin{theorem}
The \myendpoint\ mechanism
$\sfrac{35}{29}$-approximates the optimal complemented Gini index of
utilities ($\approx 1.21$).
\end{theorem}
\begin{proof}
With two or fewer agents, the \myendpoint\ mechanism returns an
optimal solution.
%in which both the utilities and the complemented
%Gini index of utilities are maximal.
Therefore we consider three or more agents.
With the \myendpoint\ mechanism, the leftmost and rightmost agents must
have utility 1, while the other agents have utility between
$\sfrac{1}{2}$ and 1. 
%Suppose there are $n$ agents ($n\geq 3$) with utilities in $[\sfrac{1}{2},1]$,
%and two of the agents have utility 1. Then
The minimum complemented Gini index of utilities
%(or equivalently the
%maximum Gini index of utilities)
is when these
$n-2$ agents have utility $\sfrac{1}{2}$. %
This occurs when 
one agent is at 0, another is at 1, the final 
$n-2$ are at $\sfrac{1}{2}$, and the facilities located at the two
endpoints. The complemented Gini
index of utilities is then $1-\frac{2(n-2)}{n(n+2)}$.
% 
% sum sum |ui - uj| = 2(n-2)
% 2n sum ui = 2n( (n-2)/2 + 2) = n(n+2)
%
This is minimized for $n=5$, when it is 
$\sfrac{29}{35}$.
Coincidently, %when 
%one agent is at 0, another is at 1, and the final 
%$n-2$ are at $\sfrac{1}{2}$
there is an optimal outcome in this case
with facilities at $\sfrac{1}{4}$ and
$\sfrac{3}{4}$, and an optimal complemented Gini index
of 1. The approximation ratio %of  the
%complemented Gini index of utilities of
is therefore $\sfrac{35}{29}$. 
\end{proof}

We now define \myendpointgamma,
a new mechanism which performs better
by {\bf truncating extreme locations}
for the two facilities using
a parameter $\gamma \in [0,\sfrac{1}{2}]$. 
For two or fewer agents, this simply applies the 
\myendpoint\ mechanism. For three or more agents,
%if $x_1$ is the leftmost agent, $x_n$ is the rightmost and
%$\gamma \in [0,\sfrac{1}{2}]$, 
%then
\myendpointgamma\ 
locates the left facility at $\mymin(\mymax(x_1,\gamma),x_n)$
and the right facility at $\mymax(x_1,\mymin(1-\gamma,x_n))$.
%The \myendpointgamma\ mechanism trivially retains
%the strategy proofness of the original. 
%It also offers a better approximation ratio than the
%original \myendpoint\ 
%mechanism. 

\begin{theorem}
The \myendpointgamma\ mechanism is strategy proof, and
$\alpha$-approximates
the optimal complemented Gini index of
utilities where
$\alpha \in [\sfrac{15}{14},\sfrac{35}{29}]$
and $\alpha$ depends on $\gamma$.
The ratio is minimized for $\gamma=\sfrac{1}{4}$
when %the mechanism
% provides a $\sfrac{15}{14}$-approximation
$\alpha=\sfrac{15}{14}$
($\approx 1.07$), and
maximized for $\gamma=0$ 
% when it provides a $\sfrac{35}{29}$-approximation
when $\alpha=\sfrac{35}{29}$
($\approx 1.21$). 
\end{theorem}
\begin{proof}
Truncating endpoints does not impact strategy proofness. 
With one or two agents, the 
\myendpointgamma\ mechanism returns an
optimal solution. 
%and perfectly equitable solution in
%which both the utilities and the complemented
%Gini index of utilities are maximal.
Therefore we consider three or more agents.
There are two cases. 
In the first case, $\gamma \leq \sfrac{1}{4}$, 
and agents must
have utility between $\sfrac{1}{2}+\gamma$ and 1.
%Suppose there are $n$ agents ($n\geq 3$) with utilities in $[\sfrac{1}{2}+\gamma,1]$.
Then
the minimum complemented Gini index of utilities
%(or equivalently the
%maximum Gini index of utilities)
is when 
$n-2$ agents have utility $\sfrac{1}{2}+\gamma$, one agent has
utility $1-\gamma$, 
and the final agent
has utility 1.
This occurs when one agent is at 0, another is
at $1-\gamma$ and
$n-2$ are at $\sfrac{1}{2}$, and the 
\myendpointgamma\ mechanism places facilities at
$\gamma$ and
$1-\gamma$. 
In this case, for fixed $n$, the complemented Gini
index of utilities increases as $\gamma$ increases
towards $\sfrac{1}{4}$.
%$\frac{(n-1)}{n(3n+1)}$. 
%
% n-2  1/2+g
% 1      1-g
% 1      1
% sum sum |ui - uj| = 2 (n-2)(1/2-g) + 2 (n-2)(1/2-2g) + 2g
% = 2(n-2) - 6g(n-2) + 2g
% = 2n-4-6ng+14g = 2n(1-3g)-4+14g
%     
% 2n sum ui = 2n (2-g+(n-2)(1/2 + g))
% = 2n (2-g+ n/2 - 1 + ng -2g) 
% = 2n (1-3g + ng + n/2)
% = n(2-6g+2ng+n)
%
For $\gamma=0$,
as shown in the previous theorem,
the approximation ratio of the
complemented Gini index of utilities
is $\sfrac{35}{29}$.
The approximation ratio decreases as $\gamma$ increases
towards $\sfrac{1}{4}$. 
For $\gamma=\sfrac{1}{4}$,
%the Gini index takes a maximum value of $\sfrac{1}{15}$
%when $n=3$. 
%The corresponding
the complemented Gini index of utilities takes a minimum of
$\sfrac{14}{15}$. 
Coincidently, when the agents are at 0, $\sfrac{1}{2}$ and
$1-\gamma$, 
there is an optimal solution
in which facilities are at $\sfrac{1}{4}+\sfrac{\gamma}{2}$ and
$\sfrac{3}{4}-\sfrac{\gamma}{2}$,
%every agent has utility $\sfrac{3}{4}-\sfrac{\gamma}{2}$, 
and the complemented Gini index of utilities 
is 1. 
The approximation ratio for $\gamma \leq \sfrac{1}{4}$
%of the complemented Gini index of utilities
is therefore in $[\sfrac{15}{14},\sfrac{35}{29}]$. 

In the second case, $\gamma \geq \sfrac{1}{4}$, 
and agents must
have utility between $1-\gamma$ and 1.
%Suppose there are $n$ agents ($n\geq 3$) with utilities in $[1-\gamma,1]$.
Then
the minimum complemented Gini index of utilities
%(or equivalently the
%maximum Gini index of utilities)
is when 
$n-1$ agents have utility $1-\gamma$, 
and the final agent
has utility 1.
This occurs when one agent is at 0, another is
at $\gamma$ and
$n-2$ are at $1$, and the 
\myendpointgamma\ mechanism places facilities at
$\gamma$ and
$1-\gamma$. 
In this case, for fixed $n$, the Gini
index of utilities increases as $\gamma$ increases
from $\sfrac{1}{4}$.
%$\frac{(n-1)}{n(3n+1)}$. 
%
% n-1  1-g
% 1      1
% sum sum |ui - uj| = 2(n-1)g
%     
% 2n sum ui = 2n ((n-1)(1-g)+1
% = 2n( n-1-ng+g+1) = 2n(n-ng+g)
%
% G = (n-1)g/n(n-ng+g)
%For $\gamma=\sfrac{1}{4}$,
%the Gini index takes a maximum value of $\sfrac{1}{15}$
%when $n=3$.
For $\gamma=\sfrac{1}{2}$,
the complemented
Gini index takes a minimum value of $\sfrac{5}{6}$
when $n=3$.
% 1/ 3(3-3/2+1/2) = 1/ 3(3/2) = 1/6
Coincidently, when the agents are at 0, 
$\gamma$ and 1
there is an optimal solution
in which facilities are at $\sfrac{\gamma}{2}$ and
$1-\sfrac{\gamma}{2}$,
%every agent has utility $1-\sfrac{\gamma}{2}$, 
and the complemented Gini index of utilities 
is 1. 
The approximation ratio for $\gamma \geq \sfrac{1}{4}$
%of the
%complemented Gini index of utilities
is therefore
%at most $\sfrac{6}{5}$ (which is less than $\sfrac{35}{29}$). 
in $[\sfrac{15}{14},\sfrac{6}{5}]$ (and $\sfrac{6}{5} <
\sfrac{35}{29}$). 
\end{proof}

\myOmit{
  It should not be surprising that 
the \myendpointgamma\ mechanism provides the best
approximation ratio
at $\gamma=\sfrac{1}{4}$. This parameter
setting is the two facility
equivalent of the \mymidornearest\ mechanism with one
facility.
If agents span $\sfrac{1}{2}$,
the \mymidornearest\ mechanism locates
the single facility uniformly in the middle of the interval
at $\sfrac{1}{2}$. 
Similarly, if agents span $[\sfrac{1}{4},\sfrac{3}{4}]$,
the \myendpointgamma\ mechanism for
$\gamma = \sfrac{1}{4}$ locates facilities
uniformly at $\sfrac{1}{4}$ and $\sfrac{3}{4}$. 
}

We now prove that no strategy proof and deterministic mechanism for two
facilities can
do better than $\sfrac{30}{29}$-approximate the optimal complemented
Gini index ($\approx 1.03$). This leaves a small gap with the
best approximation ratio of $\sfrac{15}{14}$ ($\approx 1.07$) 
achieved by the \myendpointgamma\ mechanism.
%It is an
%interesting open question to close this gap.

\begin{theorem}
No strategy proof and deterministic mechanism for two facilities can
do better than 
$\sfrac{30}{29}$-approximate the optimal complemented Gini index of
utilities ($\approx 1.03$).
\end{theorem}
\begin{proof}
Suppose a strategy proof and deterministic mechanism exists with
an approximation ratio smaller that $\sfrac{30}{29}$. 
Consider three agents, one at 0, another at $\sfrac{1}{2}$ and the
final agent at $\sfrac{3}{4}$. The optimal location of facilities that
maximizes the complemented Gini index of utilities
%(or equivalently
%minimizes the Gini index of utilities)
has one facility at
$\sfrac{1}{8}$, and the other at $\sfrac{5}{8}$ giving a complemented
Gini index of 1.
The most left that the rightmost facility can be and the
approximation ratio of the complemented Gini index be 
smaller than $\sfrac{30}{29}$ is to the right of $\sfrac{21}{26}$.
If the rightmost facility is at
$\sfrac{21}{26}$
then the minimal Gini index of utilities is when the leftmost facility
is at $\sfrac{1}{13}$ and the Gini index is $\sfrac{1}{30}$.
The complemented Gini index is then $\sfrac{29}{30}$, which
corresponds to an approximation ratio of the optimal complemented
Gini index of $\sfrac{30}{29}$.
The agent at $\sfrac{3}{4}$ therefore travels a distance greater than
$\sfrac{21}{26} - \sfrac{3}{4}$ (which is $\sfrac{3}{52}$). 

Now suppose the agent at $\sfrac{3}{4}$ reports their location as
1. The optimal solution %location of facilities 
maximizing the complemented Gini index of utilities
%(or equivalently
%minimizes the Gini index of utilities)
for the reported locations
of the agents has one facility at
$\sfrac{1}{4}$, and the other at $\sfrac{3}{4}$ giving a complemented
Gini index of 1.
The most right that the rightmost facility can be and the
approximation ratio of the complemented Gini index be 
smaller than $\sfrac{30}{29}$ is to the left of $\sfrac{21}{26}$.
If the rightmost facility is at
$\sfrac{21}{26}$
then the minimal Gini index of utilities for the reported
locations is when the leftmost facility
is at $\sfrac{5}{26}$ and the Gini index is $\sfrac{1}{30}$.
The complemented Gini index is then $\sfrac{29}{30}$, which
corresponds to an approximation ratio of the optimal complemented
Gini index of $\sfrac{30}{29}$.
Note also that the rightmost facility cannot be
to the left of $\sfrac{3}{4} - \sfrac{3}{52}$ as
this gives an approximation ratio of the complemented Gini
index greater than $\sfrac{30}{29}$. 
The agent at $\sfrac{3}{4}$ therefore travels a distance less than
$\sfrac{21}{26}-\sfrac{3}{4}$ (which is $\sfrac{3}{52}$). 
Thus, by mis-reporting their location, the agent at $\sfrac{3}{4}$
reduces
their distance from the facility from more than $\sfrac{3}{52}$ to
less than $\sfrac{3}{52}$.
%This contradicts the assumption of strategy
%proofness. 
\end{proof}

\section{Randomized Mechanisms}

Randomization is often a simple and attractive mechanism to achieve
better performance in expectation. 
For example,  the randomized \mylrm\ mechanism 
$\sfrac{3}{2}$-approximates the maximum distance any agent must
travel in expectation \cite{ptacmtec2013}.
This beats the 2-approximation lower bound that
deterministic and strategy proof mechanisms can at
best achieve. 
In addition, randomized and strategy proof mechanisms
cannot do better than $\sfrac{3}{2}$-approximate the maximum
distance\footnote{Theorem 3.4 in \cite{ptacmtec2013} shows this for
  the real line. However, the proof works for any fixed interval.}. The \mylrm\ mechanism is 
not quite as good at returning equitable solutions,
{\bf even being beaten} by a deterministic
  mechanism like \mymidornearest.
The problem with the \mylrm\ mechanism
is that, while it returns a solution
that is close to the optimal facility location in expectation,
the ex post solutions are 
at extreme locations half of the time. 

%Indeed, it does less well than
%the deterministic \mymidornearest\ mechanism. 

\begin{theorem}
  % For a facility location problem,
  The \mylrm\ mechanism
$\alpha$-approximates the optimal
complemented Gini index of utilities in expectation
with $\alpha = \sfrac{4}{3}$
for $n \leq 3$,
and %$\alpha \in [\sfrac{4}{3},\frac{2(n^2+n)}{(n^2+n+2)}]$
$\alpha \in [\sfrac{4}{3},2]$
for $n \geq
4$. 
\end{theorem}
\begin{proof}
For $n=2$, suppose one agent is at
0 and the other at $a$ with $0 \leq a \leq 1$.
With probability $\sfrac{1}{2}$, the facility is located at
$\sfrac{a}{2}$ which gives an optimal complemented Gini index of
utilities of 1.
With the remaining probability, the facility is located at
0 or $a$ which gives a sub-optimal complemented Gini index of
utilities of $1- \frac{a}{2(2-a)}$. This is minimized for $a=1$ when
the complemented Gini index of utilities is $\sfrac{1}{2}$.
The \mylrm\ mechanism thus has an expected 
complemented Gini index of utilities that is
at least $\sfrac{3}{4}$,
  compared to an optimal of 1. Hence the 
approximation ratio $\alpha = \sfrac{4}{3}$. 

  For $n=3$,
we suppose one agent is at
0, another at $a$ and the third at $b$ with $0 \leq a \leq b \leq 1$.
Without loss of generality, we suppose $2a \leq b$ (otherwise we
reflect problem). 
The optimal Gini index of utilities has the facility at 
$\sfrac{b}{2}$ giving a complemented Gini index of utilities
that is % u0 = 1-b/2
% ua = 1-b/2+a
% ub = 1-b/2
% G = 2a/3(3-3b/2+a) = 4a/3(6-3b+2a)
$1-\frac{4a}{3(6-3b+2a)}$.
With probability $\sfrac{1}{2}$, the \mylrm\ mechanism locates the facility at
$\sfrac{b}{2}$ which gives an optimal complemented Gini index of
utilities of $1-\frac{4a}{3(6-3b+2a)}$.
With probability $\sfrac{1}{4}$, the facility is located at
0 which gives a sub-optimal complemented Gini index of
utilities of 
% u0 = 1
% ua = 1-a
% ub = 1-b
% G = (a+b+b-a)/3(3-a-b) = 2b/3(3-a-b)
$1-\frac{2b}{3(3-a-b)}$.
With the remaining probability $\sfrac{1}{4}$, the facility is located at
$b$ which gives a sub-optimal complemented Gini index of
utilities of 
% u0 = 1-b
% ua = 1-b+a
% ub = 1
% G = (a+b-a+b)/3(3+a-2b) = 2b/3(3+a-2b)
$1-\frac{2b}{3(3+a-2b)}$.
The expected 
complemented Gini index of
utilities is thus
% 1 - \frac{2a}{3(6-3b+2a)} - \frac{b}{6(3-a-b) -\frac{b}{6(3+a-2b)}
% 
$1 - \frac{2a}{3(6-3b+2a)} - \frac{b}{6(3-a-b)} -\frac{b}{6(3+a-2b)}$.
The ratio of this with the optimal is minimized by $a=0$
and $b=1$ when the 
expected value is $\sfrac{3}{4}$ compared to an optimal of
$1$. Hence %the \mylrm\ mechanism 
% 5/6 /  25/36  = 30 / 25 = 6/5
%$\sfrac{4}{3}$-approximates the optimal at worst.
the approximation ratio $\alpha = \sfrac{4}{3}$.

  For $n=2k$, we consider $k$ agents at 0 and $k$ at 1. 
The optimal complemented Gini index is 1, but the expected
complemented Gini index returned by \mylrm\ is 
$\sfrac{3}{4}$. Hence $\alpha \geq \sfrac{4}{3}$. 
%
% k@0 = 1, k+1@1 = 0
% G = (k+1)/(2k+1), CG = k/(2k+1)
% k@0 = 0, k+1@1 = 1
% G = k(k+1)/(2k+1)(k+1), CG = (k+1)/(2k+1)
%
% <CG> = 1/4( 2 + (k+1)/(2k+1) + k/(2k+1))
%           = 1/4( 4k+2+k+1+k)/(2k+1) = (6k+3)/(8k+4)
%           alpha >=  (8k+4)/(6k+3)
  For $n=2k+1$, we consider $k$ agents at 0 and $k+1$ at 1. 
The optimal complemented Gini index is 1, but the expected
complemented Gini index returned by \mylrm\ is again
$\sfrac{3}{4}$. Hence $\alpha \geq \sfrac{4}{3}$.

If the facility is at the midpoint between agents,
each agent
gets at least an utility of $\sfrac{1}{2}$. The most inequitable
outcome satisfying
this constraint gives $n-1$ agents utility $\sfrac{1}{2}$ and
one agent utility 1, with a complemented
Gini index of $\frac{(n^2+1)}{(n^2+n)}$.
On the other hand,
if the facility is at one of the extreme agents,
the most inequitable
outcome has $n-1$ agents with utility $0$ and
one agent with utility 1, giving a complemented
Gini index of at least $\sfrac{1}{n}$.
Hence the expected complemented Gini index is
at least $\frac{1}{2}\frac{1}{n} + \frac{1}{2}
\frac{(n^2+1)}{(n^2+n)}$
or $\frac{(n^2+n+2)}{2(n^2+n)}$.
Therefore $\alpha \leq 
\frac{2(n^2+n)}{(n^2+n+2)} \leq 2$.
%Note that this is better than a 2-approximation. 
%
% n-1@u= 1/2, 1@1, G= (n-1)/n(n+1), CG=(n^2+1)/(n^2+n)
% n-1@u=0, 1@1l G=(n-1)/n, CG= 1/n
% <CG> >= (n^2+1)/2(n^2+n) + 1/2n = (n^2+n+2)/2(n^2+n)
% alpha <= 2(n^2+n)/(n^2+n+2)
\end{proof}

%In fact, we conjecture
%that the \mylrm\ mechanism
%provides a $\sfrac{4}{3}$-approximation for $n\geq 4$. 
%The performance
%of \mylrm\ is worse than the deterministic \mymidornearest\ mechanism. 
We also provide a lower bound on the best approximation
ratio that randomized mechanisms can achieve.

\begin{theorem}
  % For a facility location problem,
  Any randomized mechanism that is strategy proof can at best
$\sfrac{8}{7}$-approximate the optimal
complemented Gini index of utilities in expectation.  
\end{theorem}
\begin{proof}
  Suppose a strategy proof
  mechanism exist with a better approximation ratio.
Consider two agents at $\sfrac{1}{3}$ and 
 $\sfrac{2}{3}$. 
 We suppose the expected location of the
 facility is in $[0,\sfrac{1}{2}]$. The case when 
 it is in $(\sfrac{1}{2},1]$ is dual. 
  Suppose the agent at $\sfrac{2}{3}$ mis-reports their location as
  1. The optimal facility location is $\sfrac{2}{3}$, giving a
Gini index of 0.
If the facility is at $\sfrac{2}{3}-x$ for $x\in[0,\sfrac{1}{3}]$ then
the Gini index is $\sfrac{3x}{4}$. 
And if the facility is at $\sfrac{2}{3}+x$ for $x\in[0,\sfrac{1}{3}]$ then
the Gini index is $\sfrac{3x}{4}$. Hence, the expected Gini index is
$\sfrac{3}{4}$ the expected distance of the facility from
$\sfrac{2}{3}$.
% 1/3 ... 1
% u1 = 2/3+x, u2 = 2/3-x
% G = 3x/4. CG=1-3x/4
%
%
  % fac at 2/3-x for  x in [1/3,2/3]
  % u1=4/3-x, u2=2/3-x
  % G = 1/3 / (2-2x) = 1/(6-6x)
  %
  To achieve the required approximation ratio, 
 the expected complemented Gini index must be less than
 $\sfrac{7}{8}$.
  This puts the expected location of facility in the interval $(\sfrac{1}{2},\sfrac{5}{6})$.
  Hence, the agent at $\sfrac{2}{3}$ in the first
  setting has an incentive to misreport their location
  as 1. %%%%% which is a contradiction. 
% This contradicts the assumption that the
%  mechanism is strategy proof. 
\end{proof}

\section{Nash Welfare}

The solution maximizing the Nash welfare is
often considered to be an equitable compromise between the
egalitarian and utilitarian solution. 
The Nash welfare is $\sqrt[n]{\prod_i u_i}$.
We therefore also consider how well 
these strategy proof mechanisms approximate
the Nash welfare.\footnote{We again see why
  it is better to use utilities rather
  than distances. The
  product of distances suffers from the drowning effect
  of any zero distance. On the other hand, in our
  facility location
  problem,  the product of utilities 
  never suffers such
  drowning as  optimal utilities (unlike
  optimal distances) are never zero. }

\begin{theorem}
  The \myleftmost and \mymedian\ mechanisms
  have an unbounded approximation ratio of the Nash welfare,
  while the \mymidornearest\ mechanism
  $2$-approximates it. 
\end{theorem}
\begin{proof}
  Consider an agent at 0 and 1. The \myleftmost\ and
  \mymedian\ mechanisms
give one agent an utility of zero, thus giving a Nash welfare
also of zero. However,
the optimal Nash welfare is $\sfrac{1}{2}$ with
the facility located at $\sfrac{1}{2}$. 
%The approximation
%ratio is therefore not bounded. 

The \mymidornearest\ mechanism ensures that
every agent has utility of $\sfrac{1}{2}$ or greater.
The Nash welfare is therefore $\sfrac{1}{2}$ or greater.
The maximum Nash welfare is 1. Hence, the approximation
ratio is at most 2.
Consider
$n-1$ agents at 0 and a final agent at $\sfrac{1}{2}$.  
The \mymidornearest\ mechanism locates
the facility at $\sfrac{1}{2}$,
giving a Nash welfare of $\sfrac{1}{\sqrt[n]{2^{n-1}}}$.
This compares to an
optimal of $\sfrac{1}{\sqrt[n]{2}}$ with the facility at
$0$. This gives an approximation
ratio of $\sqrt[n]{2^{n-2}}$. This approaches 2 from below
as $n$ goes to infinity. 
Hence, the approximation ratio is at least 2. 
\end{proof}

We next prove that any deterministic mechanism
can at best approximate the optimal Nash welfare.

\myOmit{
\begin{theorem}
  No deterministic
and strategy proof mechanism
  can do better than $\frac{5}{2 \sqrt{6}}$-approximate the Nash
  welfare ($\approx 1.02$). 
\end{theorem}
\begin{proof}
Consider an agent at 0 and $\sfrac{1}{2}$.
The optimal Nash welfare is $\sfrac{3}{4}$ with the
facility at $\sfrac{1}{4}$. Suppose there is a deterministic
and strategy proof mechanism with a better approximation
ratio. Then
the facility must be to the left of $\sfrac{2}{5}$
as the Nash welfare
is $\sfrac{3 \sqrt{6}}{10}$
when the facility is at $\sfrac{2}{5}$, and this
is $\frac{2\sqrt{6}}{5}$ times the optimal Nash welfare. If the
facility is to the right of $\sfrac{2}{5}$,
the Nash welfare is even less,
and fails to achieve the approximation ratio. 
Suppose the agent at $\sfrac{1}{2}$ reports their location
as 1. The optimal Nash welfare given this report
is $\sfrac{1}{2}$ with the
facility at $\sfrac{1}{2}$. To achieve the
approximation ratio on
the Nash welfare, the facility must be in the
interval $(\sfrac{2}{5},\sfrac{3}{5})$. But this puts the facility
strictly closer to $\sfrac{1}{2}$ contradicting the assumption
that the mechanism is strategy proof. 
\end{proof}

If we additionally insist on Pareto efficiency, we are able
to increase the lower bound modestly.

}

\begin{theorem}
  A deterministic mechanism 
  that is anonymous, Pareto efficient
  and strategy proof at best
  $\frac{3}{2\sqrt{2}}$-approximates the Nash welfare ($\approx 1.06$). \end{theorem}
\begin{proof}
Such a mechanism is a median mechanism with $n-1$ phantoms.
Consider $n=2$ and a phantom at $a$. 
Suppose $a \leq \sfrac{1}{2}$.
The case with $a \geq \sfrac{1}{2}$ is dual. 
Consider one agent at 0, and another at 
$a$.
The median mechanism locates the facility at $a$ giving 
a Nash welfare of
$\sqrt{(1-a)}$. 
This compares to an optimal of $1-\sfrac{a}{2}$. 
Therefore the approximation ratio is
$\frac{(1-\sfrac{a}{2})}{\sqrt{(1-a)}}$ which is in $[1, \frac{3}{2\sqrt{2}}]$
for $0 \leq a \leq \sfrac{1}{2}$. 
Consider one agent at $a$, and another at 
1. The median mechanism again locates the
facility at $a$ gving
a Nash welfare of
$\sqrt{a}$. 
This again compares to an optimal of $\frac{(a+1)}{2}$. 
Therefore the approximation ratio is
$\frac{(a+1)}{2\sqrt{a}}$ which is in $[\frac{3}{2\sqrt{2}},\infty)$ for $0 \leq a \leq \sfrac{1}{2}$.
Over the two scenarios, the best approximation
ratio that can be achieved is $\frac{3}{2\sqrt{2}}$. 
\end{proof}

%\section{Randomized mechanisms for Nash welfare}

It is an interesting question to close
the gap between this lower bound
and the 2-approximability
provided by the \mymidornearest\ mechanism. 

Can randomized mechanisms do better?
Again the \mylrm\ mechanism is an obvious
candidate %to do better
given its optimality
at minimizing the maximum distance. 
It is, however, 
less good at approximating the Nash welfare. Indeed,
it is again beaten by deterministic mechanisms like
the \mymidornearest\ mechanism.

\begin{theorem}
  The \mylrm\ mechanism $\alpha$-approximates the Nash welfare
  in expectation with $\alpha=2$ for $n=2$,
  $\alpha=\frac{4}{\sqrt[3]{4}}$ ($\approx 2.52$) for $n=3$ and $\alpha=4$
for $n\geq 4$. 
\end{theorem}
\begin{proof}
  For $n=2$ and $n=3$, we perform a similar case analysis to
the  proof of Theorem 9. 
For $n\geq 4$, consider the %case in which the
  % mylrm\
%  mechanism
%  locates the
facility at $\sfrac{(x_1+x_n)}{2}$.
  The utility of every agent is at least $\sfrac{1}{2}$.
  Hence the Nash welfare is also at least $\sfrac{1}{2}$.
  Since this occurs with probability $\sfrac{1}{2}$, 
  this contributes at least $\sfrac{1}{4}$ to the expected
  Nash welfare. As the optimal Nash welfare is at most 1, the
  approximation ratio is at most 4.
    Consider $n-1$ agents at 0 ($n>1$), and one final agent at 1.
  The expected Nash welfare of the probability
  distribution of solutions returned
  by the \mylrm\ mechanism is $\sfrac{1}{4}$. \myOmit{ as
  the Nash welfare when the facility is located at 0 or 1 is
  zero, and at $\sfrac{1}{2}$ is $\sfrac{1}{2}$.}
  % 0, a
  % 1, 1-a
  % NW = root(n,a^{n-1}(1-a))
  % d (UV) = dUV + UdV
  % d NW = (n-1)a^{n-2} (1-a) - a^{n-1) = 0
  % (n-1) (1-a) - a = 0
  % n-1-na = 0
  % a = (n-1)/n
  % a = 1-1/n
The optimal Nash welfare
\myOmit{is, however, $\sqrt[n]{\frac{(1-\sfrac{1}{n})^{n-1}}{n}}$
  with the facility at $\sfrac{1}{n}$. This}
  approaches 1 from below as $n$ goes to infinity.
  Therefore the approximation ratio is at least 4. 
\end{proof}

Finally, we give approximability results for two facilities.

\begin{theorem}
The \myendpoint\ mechanism
$2$-approximates the optimal Nash welfare. 
\end{theorem}
\begin{proof}
With two or more agents, 
the \myendpoint\ mechanism, the leftmost and rightmost agents must
have utility 1, while the other agents have utility between
$\sfrac{1}{2}$ and 1. 
%Suppose there are $n$ agents ($n\geq 3$) with utilities in $[\sfrac{1}{2},1]$,
%and two of the agents have utility 1. Then
The minimum Nash welfare
is when these
$n-2$ agents have utility $\sfrac{1}{2}$
and the Nash welfare is %$\frac{1}{\sqrt[n]{2^{n-2}}}$
${1}/{\sqrt[n]{2^{n-2}}}$. This tends
  to $\sfrac{1}{2}$ for above as $n$ goes to infinity. 
  Hence the approximation ratio is at most 2.
Suppose one agent is at 0, another is at 1, the final 
$n-2$ are at $\sfrac{1}{2}$, and the facilities located at the two
endpoints. The Nash welfare
 of the solution returned by the mechanism is
 % $\frac{1}{\sqrt[n]{2^{n-2}}}$.
${1}/{\sqrt[n]{2^{n-2}}}$.
The optimal Nash welfare 
% NW = a^2(1-a)^{n-2}
% dNW/da = 2a(1-a)^{n-2) -(n-2)a^2(1-a)^{n-1} = 0
%  2(1-a) - (n-2)a = 0, a=1/n
%  $\sqrt[n]{\frac{(1-\sfrac{1}{n})^{n-2}}{n^2}}$
$\sqrt[n]{{(1-\sfrac{1}{n})^{n-2}}/{n^2}}$
with the facility at $\sfrac{1}{n}$. This
approaches 1 from below as $n$ goes to infinity.
Hence the approximation ratio is at least 2.
\end{proof}

Using the same examples as in the proof of Theorem 8, we
can also show that any strategy proof and deterministic mechanism
for two facilities cannot do better than $\sfrac{169}{168}$-approximate
the Nash welfare ($\approx 1.01$).

\section{Related Work}

Several recent surveys
summarize the considerable literature on mechanism
design for facility location \cite{faclocsurvey,ijcai2021-596}. 
Beginning with Procaccia and Tennholtz
\shortcite{approxmechdesign2}, most studies of strategy
proof mechanisms for facility location have focused on
approximating the total and maximum distance
(e.g. \cite{egktps2011,ft2010,ft2013,coordmedian,proportional,flplimit2,zhang2014,wijcai2022}).
%For
%example, the \mymedian\ mechanism returns the optimal total distance, and
%-approximates the maximum distance, and no other deterministic and
%strategy proof mechanism can do better 
%\cite{ptacmtec2013}. 
Indeed, % in a recent survey of mechanism design for facility
% location, % problems,
one of these recent surveys describes
designing strategy proof mechanisms for facility
location which approximate
well the total or maximum distance that agents travel
as the ``classic setting'' %for approximate mechanism design
\cite{ijcai2021-596}.

One of the simplest measures of equitability in facility location is 
the maximum distance agents travel.
% That is, we consider the objective of minimizing
%the maximum distance. %, $\mymax_i d_i$. 
Marsch and Schilling
\shortcite{equityflp}
claim that this
``is the earliest and most frequently used measure that has an
equity component'' in facility location problems, that
``it has long been used as a more equitable alternative to the 
$p$-median problem which minimizes [total] travel distance'', 
and that it ``quantifies the popular Rawlsian 
criteria of equity which seeks to improve as much as possible
those who are worst-off''.
%Indeed, maximum distance has been
%considered in approximate mechanism design for
%facility location from the very start.
%(e.g. 
%Procaccia and Tennenholtz 
%\cite{ptacmtec2013}).

\myOmit{
Related to maximum distance is minimum happiness. 
The ``happiness''
of agent $i$ is
$h_i = 1 - \sfrac{d_i}{d_{max}^i}$ where $d_{max}^i$ is the maximum
possible distance agent $i$ can travel
\cite{flprevisit,flpdesire}. %For instance,
With
agents and facilities 
constrained to %the interval
$[0,1]$, $d_{max}^i =\mymax(x_i,1-x_i)$. 
This normalization changes the approximation ratios achievable.
% compared to approximating just the maximum distance.
For example,
the \mymedian\ mechanism, which 2-approximates the maximum distance, 
% agents travel,
does not bound the minimum happiness of
any agent. }

Another common measure of equitability is variance. For example,
Maimon  \shortcite{maimon} develops an algorithm 
to locate a facility on a tree network minimizing
the variance in distances agents travel. 
Procaccia {\it et al.}
\shortcite{pwzaaai2018} study a different use of variance, 
exploring the tradeoff in randomized mechanisms between variance 
in the distribution of the location a facility and
the approximation ratio of the optimal total or maximum distance
agents travel. 
Other simple measures of equitability are the range and
absolute deviation in distances
agents travel \cite{equityflp}.
For example, Berman and Kaplan \shortcite{rangeflp}
argue that the absolute deviation is ``a natural
measure of the equity'' 
of facility location problems and provide
an efficient algorithm to compute the
location of a facility on a general network to minimize
this measure. 

There are other indices of inequality besides the Gini
index. 
For example, a common measure of income inequality is the Hoover index (also known
as the Robin Hood or Schutz index), and this has been applied
to facility location \cite{equiflp}. 
As a second example, the Atkinson index has been used
in social choice settings such as resource allocation
\cite{atkinson}. 
It would be interesting to design mechanisms
approximating such indices. 

Lam {\it et al.} \shortcite{DBLP:journals/corr/abs-2310-04102}
study mechanisms for facility location that
maximize the product of
agents' utilities. They
give a polynomial time approximation algorithm to compute the facility
location maximizing this product, and prove results suggesting
that this achieves a good balance between fairness and efficiency.
They also prove
that no deterministic and strategy-proof mechanism provides a bounded approximation
of the optimal product of utilities.
Here we show that the optimal
Nash welfare (the $n$th root of this
product) can, on the other hand, be approximated well. 

\section{Conclusions}

\begin{table}[htb]
   \hspace{0.3cm}
 % \begin{center}
    \begin{tabular}{|c|c|c|} \hline
 & {\bf 1-G} & {\bf Nash} \\ \hline
1 facility, deterministic &  &  \\ \hline
lower bound & $\pmb{\sfrac{6}{5}}$  & $\pmb{\frac{3}{2\sqrt{2}}}$ \\
                   \myleftmost   & $\pmb{n}$ & $\pmb{\infty}$
                   \\
                   \mymedian & {\bf 2} & $\pmb{\infty}$  \\
                   \mymidornearest, $n<4$ &
                                     $\pmb{\sfrac{6}{5}}$& {\bf 2}  \\
                   \mymidornearest, $n\geq 4$ &
                                     $\pmb{\frac{n^2+n}{n^2+1}}$& {\bf 2}  \\
                   \hline
1 facility, randomized &   &  \\ \hline
lower bound &   $\pmb{\sfrac{8}{7}}$ &  -- \\
                   \mylrm, $n=2$       &  $\pmb{\sfrac{4}{3}}$& $\pmb{2}$ \\
                   \mylrm, $n=3$       &  $\pmb{\sfrac{4}{3}}$ &
                                                                 $\pmb{\frac{4}{\sqrt[3]{4}}}$  \\
                   \mylrm, $n\geq 4$ & $\pmb{[\sfrac{4}{3},2]}$ &
                                                                  $\pmb{4}$
                   \\
%                   \mylrmt & -- & $\pmb{[2,\sfrac{12}{5}]}$ \\
   \hline 
2 facilities, deterministic &   & \\ \hline
                   lower bound & $\pmb{\sfrac{30}{29}}$ &
                                                          $\pmb{\sfrac{169}{168}}$
                   \\ 
                   \myendpoint              & $\pmb{\sfrac{35}{29}}$ & $\pmb{2}$
                   \\ 
                   \myendpointgamma              &
                                                   $\pmb{[\sfrac{15}{14},\sfrac{35}{29}]}$
                   & -- \\                    \myendpointgamma, $\gamma=\sfrac{1}{4}$              &
                                                   $\pmb{\sfrac{15}{14}}$
                   & -- \\                    \myendpointgamma, $\gamma=\sfrac{1}{2}$              &
                                                   $\pmb{\sfrac{6}{5}}$
                   & -- \\ \hline
                   
\end{tabular}
%\end{center}
\caption{Summary of approximation ratios of
  the complemented
  Gini index of utilities ({\bf 1-G}) and of the Nash welfare ({\bf Nash}). 
%  {\bf Bold} font for those
%results proved here. $\gamma \in [0,\sfrac{1}{2}]$. 
}
\end{table}

We have proposed designing mechanisms that
approximate well equitability. 
%For the facility location
%problem, we argued why the Gini index of agent utilities
%is a better measure of equitability
%than the Gini index of distances that agents travel. % to the
%nearest facility.
We first proved an impossibility result
that strategy proof mechanism %for one or more facilities
cannot bound the approximation ratio
of the optimal Gini index of utilities. % (or of distances).
We instead turn the problem on its head
by considering approximation ratios of the
complemented Gini index. % of utilities.
This focuses on problems without perfect or near
to perfect solutions
where not all agents can be equi-distant from the facility.
As Nash welfare is often considered to be an equitable compromise
between egalitarian and utilitarian solutions, we also
considered how well mechanisms optimize the Nash welfare.
Our results are summarized in Table 1. They demonstrate
that it is possible to design strategy proof mechanisms with near
optimal equitable solutions.
For a single facility, we identified a deterministic and strategy
proof
mechanism with an optimal approximation ratio.
For two facilities, we designed a new strategy proof mechanism
with a near optimal approximation ratio.
\myOmit{For both deterministic and randomized mechanisms for a single facility, we identified
mechanisms that bound the approximation ratio of this
objective. In the case of deterministic mechanisms, we 
identify a mechanism with an optimal ratio.
We then extended results to multiple facilities. For instance,
we
proposed a new strategy proof mechanism with a better approximation
ratio for two facilities than the \myendpoint\ mechanism,
the only strategy proof mechanism with a bounded 
approximation ratio
of the optimal minimum utility or maximum distance.
This last results demonstrate the importance of censoring
extreme solutions when looking for equitability. 
}

\section*{Acknowledgements}

The author is supported by the ARC through Laureate Fellowship FL200100204.

\bibliographystyle{named}
%\bibliography{/Users/twalsh/Documents/biblio/a-z,/Users/twalsh/Documents/biblio/a-z2,/Users/twalsh/Documents/biblio/pub,/Users/twalsh/Documents/x1biblio/pub2}
\bibliography{/Users/z3193295/Documents/biblio/a-z2,/Users/z3193295/Documents/biblio/pub2}
%\bibliography{/home/tw/biblio/a-z,/home/tw/biblio/a-z2,/home/tw/biblio/pub,/home/tw/biblio/pub2}

\end{document}